\documentclass[reqno]{amsart}
\usepackage{amsmath}
\usepackage{amsfonts,amscd, amssymb}
\usepackage{hyperref}
\usepackage{xcolor}
\usepackage{tikz-cd}
\usepackage{amsthm,latexsym,euscript,amscd}

\usetikzlibrary{matrix}
\newtheorem{theorem}{Theorem}[section]

\newtheorem{corollary}[theorem]{Corollary}
\newtheorem{proposition}[theorem]{Proposition}
\theoremstyle{remark}
\newtheorem{remark}[theorem]{Remark}
\theoremstyle{definition}
\newtheorem{definition}[theorem]{Definition}

\numberwithin{equation}{section} 
\makeatother

\DeclareMathOperator{\Cdb}{{\mathbb C}}

\DeclareMathOperator{\Rdb}{{\mathbb R}}

\DeclareMathOperator{\Al}{{\mathcal A}}
\DeclareMathOperator{\Bl}{{\mathcal B}}

\DeclareMathOperator{\Hl}{{\mathcal H}}

\newcommand{\ball}[1]{\operatorname{Ball}(#1)}

\newcommand{\norm}[1]{\left\Vert#1\right\Vert}

\numberwithin{equation}{section}

\makeatletter
\def\blfootnote{\xdef\@thefnmark{}\@footnotetext}

\begin{document}


\title[Suzuki Type Estimates and Lie-Trotter Formulas 
in JB-Algebras]
{Suzuki Type Estimates for Exponentiated Sums 
	\\
	and Generalized Lie-Trotter Formulas \\ 
in JB-Algebras}


\author{Sarah Chehade}
\address{Quantum Computational Science Group, Computational Sciences and Engineering Division, Oak Ridge National Laboratory Oak Ridge, TN 37831 USA }
\email{chehades@ornl.gov}
\author{Shuzhou Wang}
\address{Department of Mathematics, University of Georgia, Athens, GA 30602 USA}
\email{szwang@uga.edu}

\author{Zhenhua Wang}
\address{Department of Mathematics, University of Georgia, Athens, GA 30602 USA}
\address{Current Address: Department of Physics, Chemistry and Mathematics, Alabama A\&M University, Normal, AL 35762-4900}
\email{ezrawang@uga.edu}


\date{}
\begin{abstract}

Lie-Trotter-Suzuki product formulas are ubiquitous in quantum mechanics, computing, and simulations. Approximating exponentiated sums with such formulas are investigated in the JB-algebraic setting. We show that the Suzuki type approximation for exponentiated sums 
holds in JB-algebras, we give explicit estimation formulas, 
and we deduce three generalizations of Lie-Trotter formulas for 
arbitrary number elements in such algebras. 
We also extended the Lie-Trotter formulas  
in a Jordan Banach algebra from three elements  
to an arbitrary number of elements. 
\end{abstract}
\maketitle

\section{Introduction}
Lie-Trotter product formulas provide an elegant approach to approximate exponential sums with a product of exponentials. This technique is vastly used in quantum mechanics whenever a Hamiltonian is written as a sum of operators that do not commute. Such a Hamiltonian operator describes a system's total energy. 

In 1934, Jordan axiomatized quantum theory through what are known as Jordan algebras \cite{jordan1993algebraic}. These \emph{non-associative} algebras capture quantum mechanics in a more general way. It was conjectured that this formulation of quantum mechanics would also apply to relativity and nuclear phenomena \cite{jordan1993algebraic}. Alternatively, Jordan algebras are the abstraction and generalization of real/complex quantum theory. For example, in standard quantum theory, a product of 2 observables (self adjoint operators) need not be an observable. However a \emph{Jordan} product of 2 observables is indeed an observable. As a motivation for this line of research, observables in a quantum system constitute a JB-algebra which is non-associative, therefore JB-algebras were considered as natural objects of study for quantum system, where Lie-Trotter product formulas are heavily used. 

Though widely used, product formula error estimates suffer from a lack of understanding, difficulty in its computation, and scalability. 
In the 1976 paper \cite{Suzuki1976}, Suzuki established an estimate
for the product formula approximation for exponentiated sums of an arbitrary number of elements in a Banach algebra, generalizing the Lie-Trotter formula for just two elements. 
Later, in a 1985 paper \cite{Suzuki1985}, he established another estimate for exponentiated sums in Banach algebras, improving the one in \cite{Suzuki1976}. 
Over the years, error estimates of Lie-Trotter product formulas were formulated for bounded matrix algebras, Hilbert spaces, Banach spaces, and operator ideals along with their unbounded counterparts.

This paper introduces Lie-Trotter product formula error estimates for particular Jordan algebras. It should be pointed out that because of non-associativity, the connotations of Lie-Trotter or generalized Suzuki product formulas here are different from the classical situation in Banach algebras, where associativity holds. Even if all elements in JB-algebras are Hilbert space operators, our result in Theorem \ref{Tjsuzuki1} appears to be new in the literature.

In Section \ref{section: error}, we prove three error estimates of Lie-Trotter product formulas in the JB-algebraic setting, which is a subclass of Jordan algebras. In addition, we deduce three generalizations of the Lie-Trotter product formulas for an arbitrary finite number of elements in such algebras and we give three explicit estimation formulas: one of order $\displaystyle \frac{1}{n}$, and two of order $\displaystyle \frac{1}{n^2}$. The latter 2 are a quadratic speed up of the original estimate \cite{Suzuki1976}. Despite an arbitrary number of elements in a non-associative JB-algebra, 
it is surprising that Suzuki type estimates still hold for such algebras.

More recently, Escolano, Peralta, and Villena \cite{Peralta2023} extended the Lie-Trotter formula of two and three elements to Jordan-Banach algebras over complex numbers. In Section \ref{section: general}, we generalize their results to an exponentiated finite arbitrary sum of elements in Jordan-Banach algebras. It is not clear to us and an interesting question whether Suzuki's error estimates are also valid for these general Jordan-Banach algebras. 

\section{Preliminaries} 

In this section, we give some background on Jordan Banach algebra, JB-algebra and fix the notation. For more information, we refer the reader to \cite{Alfsen2003, hanche1984jordan}.

\begin{definition}\label{DefJa}
	A {\bf Jordan algebra} $\Al$ over real number is a vector space $\Al$ over $\Rdb$ equipped with a  bilinear product $\circ$ that satisfies the following identities:
	$$A\circ B =B\circ A, \,\ \,\ (A^2\circ B)\circ A=A^2\circ (B\circ A),$$
where $\displaystyle A^2$ means $\displaystyle A\circ A.$ 
\end{definition}
Any associative algebra $\Al$ has an underlying Jordan algebra structure with Jordan product given by 
	$$A\circ B=(AB+BA)/2.$$
	Jordan subalgebras of such underlying Jordan algebras are called {\bf special}. 
	
\begin{definition}\label{DefwJB}
	A real or complex {\bf Jordan-Banach algebra} is a Jordan algebra $\Al$ over $\Rdb$ or $\Cdb$ with a complete norm satisfying $\norm{A\circ B}\leq \norm{A}\norm{B}$ for $A, B\in \Al.$ 
\end{definition}


\begin{definition}\label{DefJB}
	A {\bf JB-algebra} is a Jordan algebra ${\mathcal A}$ over ${\mathbb R}$ with a complete norm satisfying the following conditions for $A, B\in {\mathcal A}:$ 
	$$
	\Vert A\circ B\Vert \leq \Vert A\Vert \Vert B\Vert,~~\Vert A^2\Vert=\Vert A\Vert^2,~~{\rm and}~~\Vert A^2\Vert \leq \Vert A^2+B^2\Vert.	
	$$
\end{definition}
Note that every JB-algebra is a real Jordan Banach algebra but the converse is not true. 

The set of bounded self adjoint operators on a Hilbert space $H$, denoted by $B(H)_{sa}$, is an important object in physics, as it is the set of observables in a quantum mechanical system. Moreover, it is a JB-algebra, i.e., it is not an associative algebra.

\begin{definition}
	Let $\Al$ be a Jordan algebra and $A, B, C \in \Al$.  Define 	
	\begin{align}\label{JI}
	\{ABC\}&:=(A\circ B)\circ C+(B\circ C)\circ A-(A\circ C)\circ B.
			\end{align}
\end{definition}
Note that in some literature for $C=A$, $\{ABA\}$ is denoted by $U_A(B)$ instead. This linear map is extremely important for many reasons. For example, a non-trivial fact is if $\Al$ is a JB-algebra and $B$ is positive, then so is $\{ABA\}$. Another example of its power is its use in Theorem \ref{Tjsuzuki2}. 
In general, $ABA$ is meaningless unless $\Al$ is special, in which case $\{ABA\}=ABA.$

Recall for an element $A$ in a JB-algebra $\Al$ and a continuous function $f$ on the spectrum of $A$, 
$f(A)$ is defined by functional calculus in JB-algebras (see e.g. \cite[Proposition 1.21]{Alfsen2003}).  

\section{Suzuki estimations }\label{section: error}
For $\displaystyle A_1, A_2,\cdots, A_m$ in a unital JB-algebra or Jordan-Banach algebra, let
\begin{align*}
	g_n\left(\{A_j\}\right)&=\left(\left[\left(\exp\left(\frac{A_1}{n}\right)\circ \exp\left(\frac{A_2}{n}\right)\right)\circ\cdots \right]\circ \exp\left(\frac{A_m}{n}\right)\right)^n,\\
	f_n\left(\{A_j\}\right)&=\left(\left\{\exp\left(\frac{A_m}{2n}\right)\cdots\left\{\exp\left(\frac{A_2}{2n}\right)\exp\left(\frac{A_1}{n}\right)\exp\left(\frac{A_2}{2n}\right) \right\}\cdots	\exp\left(\frac{A_m}{2n}\right)\right\}\right)^n.	
\end{align*}

The following theorem extends Suzuki's estimations for exponential of sums \\ 
in \cite[Theorem 3]{Suzuki1976} to the non-associative setting and with the improved order $\displaystyle \frac{1}{n^2}$. 
\begin{theorem}\label{Tjsuzuki1}
	For any finite number of elements $\displaystyle A_1, A_2,\cdots, A_m$ in a unital JB-algebra $\Al,$
	
	\begin{align*}
		\left\Vert \exp \left( \sum_{j=1}^m A_j\right)-g_n(\{A_j\})\right \Vert	
		\leq
		\dfrac{1}{3n^2}\left( \sum_{j=1}^m \Vert A_j\Vert\right)^3\exp\left(\sum_{j=1}^m \Vert A_j\Vert\right).	
	\end{align*}
\end{theorem}
\begin{proof}
	Let
	\begin{align}
		C&=\exp\left(\frac{1}{n}\sum_{j=1}^m A_j\right),\label{NotionC}\\
		D&=\left[\left(\exp\left(\frac{A_1}{n}\right)\circ \exp\left(\frac{A_2}{n}\right)\right)\circ\cdots \right]\circ \exp\left(\frac{A_m}{n}\right)\nonumber.
	\end{align}
	Then,
	\begin{align*}
		\Vert C\Vert\leq\exp \left(\frac{1}{n}\sum_{j=1}^m \Vert A_j\Vert\right)\,\ \mbox{and}\,\  
		\Vert D\Vert\leq \exp \left(\frac{1}{n}\sum_{j=1}^m \Vert A_j\Vert\right).
	\end{align*}
	
	Let $\Bl$ be the JB-subalgebra generated by $\{C, D, I\}.$ By Shirshov-Cohen theorem for JB-algebras (see e.g. \cite[Theorem 7.2.5]{hanche1984jordan}), there exists an isometric Jordan homomorphism $\pi: \Bl \to B(\Hl)_{sa},$ where $\Hl$ is a complex Hilbert space. Therefore,
	\begin{align}
		s&=\left\Vert \exp \left( \sum_{j=1}^m A_j\right)-g_n(\{A_j\})\right \Vert	=\left\Vert C^n-D^n \right\Vert \nonumber \\
		&=\vert \pi(C^n)-\pi(D^n)\Vert=\Vert \pi(C)^n-\pi(D)^n\Vert \nonumber \\
		&\leq \Vert \pi(C)-\pi(D)\Vert\left(\Vert \pi(C)\Vert^{n-1}+\Vert \pi(C)\Vert^{n-2}\Vert \pi(D)\Vert+\cdots+ \Vert \pi(C) \Vert \Vert \pi(D)\Vert^{n-2} +\Vert \pi(D)\Vert^{n-1}\right) \nonumber\\
		&=\Vert C-D\Vert\left(\Vert C\Vert^{n-1}+\Vert C\Vert^{n-2}\Vert D\Vert+\cdots+\Vert D\Vert^{n-1}\right)\nonumber \\
		&\leq  n \Vert C-D\Vert\left(\max\{\Vert C\Vert, \Vert D\Vert\}  \right)^{n-1} \nonumber \\
		&\leq n \Vert C-D\Vert \exp \left(\frac{n-1}{n}\sum_{j=1}^m \Vert A_j\Vert\right)\nonumber \\
		&\leq n \left(\Vert C-F\Vert +\Vert D-F\Vert\right)\exp \left(\frac{n-1}{n}\sum_{j=1}^m \Vert A_j\Vert\right), \label{GHnbound3} 
	\end{align}
	where
	\begin{align}
		F=	I+\left( \sum_{j=1}^m \frac{A_j}{n}\right)+\frac{1}{2!}\left( \sum\limits_{j=1}^m \dfrac{A_j}{n}\right)^2. \label{NotionF}
	\end{align}
	We claim that the degree 2 Taylor polynomial of $\displaystyle C$ and $\displaystyle D$ is $K.$ It is obvious for $\displaystyle C$. For any $2\leq k\leq m,$ let
	\begin{align*}
		D_k=\left[\left(\exp\left(\frac{A_1}{n}\right)\circ \exp\left(\frac{A_2}{n}\right)\right)\circ\cdots \right]\circ \exp\left(\frac{A_k}{n}\right).	
	\end{align*}
	Denote
	\begin{align}
		F_k=	I+\left( \sum_{j=1}^k \frac{A_j}{n}\right)+\frac{1}{2!}\left( \sum\limits_{j=1}^k \dfrac{A_j}{n}\right)^2. \label{termK}
	\end{align}
	The degree 2 Taylor polynomial of $\displaystyle D_2$ is 
	\begin{align*}
		F_2=I+\left(\frac{A_1}{n}+\frac{A_2}{n}\right)+\dfrac{1}{2!}\left(\frac{A_1}{n}+\frac{A_2}{n}\right)^2	
	\end{align*}
	Since $\displaystyle  D_3=D_2 \circ \exp\left(\frac{A_3}{n}\right),$ a similar computation shows that the degree 2 Taylor polynomial of $\displaystyle D_3$ is 
	\begin{align*}
		I+\left(\frac{A_1}{n}+\frac{A_2}{n}\right)+\dfrac{1}{2!}\left(\frac{A_1}{n}+\frac{A_2}{n}\right)^2+\frac{A_3}{n}+\left(\frac{A_1}{n}+\frac{A_2}{n}\right)\circ \frac{A_3}{n}+	\frac{\left(\frac{A_3}{n}\right)^2}{2},
	\end{align*}
	which is  
	\begin{align*}
		F_3=I+\left(\frac{A_1}{n}+\frac{A_2}{n}+\frac{A_3}{n}\right)+\dfrac{1}{2!}\left(\frac{A_1}{n}+\frac{A_2}{n}+\frac{A_3}{n}\right)^2.
	\end{align*}
	By induction, the degree 2 Taylor polynomial of $D_m$ is $F_m.$
	
	Therefore, 
	\begin{align}
		\Vert C-F\Vert&=\left\Vert \sum_{k=3}^{\infty}\left(\frac{A_1+A_2\cdots+A_m}{n}\right)^k\frac{1}{k!}\right\Vert \nonumber \\
		&\leq \sum_{k=3}^{\infty}\left(\frac{\Vert A_1\Vert+\Vert A_2\Vert\cdots+\Vert A_m\Vert}{n}\right)^k\frac{1}{k!}\nonumber \\ 
		&= \exp \left(\frac{1}{n}\sum_{j=1}^m \Vert A_j\Vert \right)-\left(I+\dfrac{\sum_{j=1}^m \Vert A_j\Vert}{n}+\dfrac{\left(\sum_{j=1}^m \Vert A_j\Vert\right)^2}{2n^2}\right) \nonumber \\\
		&\leq \frac{1}{3!\cdot n^3}\left(\sum_{j=1}^m \Vert A_j\Vert \right)^3\exp \left(\frac{1}{n}\sum_{j=1}^m \Vert A_j\Vert \right) \label{TaylorC} 
	\end{align}
	\begin{align}
		\Vert D-F\Vert&=\Vert D_m-F_m\Vert \nonumber \\ 
  &=\left\Vert \sum_{k_1+\cdots+k_m=3}^{\infty} \dfrac{\left[(A_1^{k_1}\circ A_2^{k_2})\circ \cdots\right]\circ A_m^{k_m} }{n^{k_1+k_2+\cdots+k_m}\cdot k_1!\cdot k_2!\cdots k_m!}\right\Vert \nonumber \\ 
		&\leq \sum_{k_1+\cdots+k_m=3}^{\infty} \dfrac{\Vert A_1\Vert^{k_1} \Vert A_2\Vert^{k_2}\cdots \Vert A_m \Vert ^{k_m} }{n^{k_1+k_2+\cdots+k_m}\cdot k_1!\cdot k_2!\cdots k_m!} \nonumber \\ 
		&= \exp \left(\frac{1}{n}\sum_{j=1}^m \Vert A_j\Vert \right)-\left(I+\dfrac{\sum_{j=1}^m \Vert A_j\Vert}{n}+\dfrac{\left(\sum_{j=1}^m \Vert A_j\Vert\right)^2}{2n^2}\right) \nonumber \\ 
		&\leq \frac{1}{3!\cdot n^3}\left(\sum_{j=1}^m \Vert A_j\Vert \right)^3\exp \left(\frac{1}{n}\sum_{j=1}^m \Vert A_j\Vert \right) \label{TaylorD}
	\end{align}
	where the inequalities (\ref{TaylorC}) and (\ref{TaylorD}) follow from \cite[Theorem 1]{Suzuki1976}.
	
	Combining (\ref{GHnbound3}), (\ref{TaylorC}) and (\ref{TaylorD}), we get desired result.    
\end{proof}

As a corollary, we have the following generalized Lie-Trotter formula for an 
arbitrary number of elements in a JB-algebra. 

\begin{corollary}
	For any finite number of elements $A_1, A_2,\cdots, A_m$ in a JB-algebra $\Al,$
	\begin{align*}
		\lim\limits_{n\to \infty}\left\{\left[\left(\exp\left(\frac{A_1}{n}\right)\circ \exp\left(\frac{A_2}{n}\right)\right)\circ\cdots \right]\circ \exp\left(\frac{A_m}{n}\right)\right\}^n=\exp\left(\sum_{j=1}^m A_j\right).
	\end{align*}
	If $\displaystyle m=2,$ then it is reduced to the following Lie-Trotter formula for JB-algebra:
	\begin{align*}
		\exp(A+B)=\lim\limits_{n\to \infty}\left(\exp\left(\frac{A}{n}\right)\circ \exp\left(\frac{B}{n}\right)\right)^n	
	\end{align*}
	for any elements $A, B$ in $\Al.$	
\end{corollary}
The following result generalizes the Suzuki symmetrized approximations for the exponentiated sums 
in \cite[Formula 3, Equation (1.15)\footnote{Note that in Suzuki's original paper, there is a typo. The power should be a 3 instead of a 2. See Formula 2 (Equation 1.10) there for the correct powers.}]{Suzuki1985} to JB-algebras with different connotations. 

\begin{theorem}\label{Tjsuzuki2}
	For any finite number of elements $A_1, A_2,\cdots, A_m$ in a JB-algebra $\Al,$
	\begin{enumerate}
		\item[(i)] $\displaystyle \left\Vert \exp \left( \sum_{j=1}^m A_j\right)-f_n(\{A_j\})\right \Vert	
		\leq
		\dfrac{3^{m-1}+1}{6n^2}\left( \sum_{j=1}^m \Vert A_j\Vert\right)^3\exp\left(\sum_{j=1}^m \Vert A_j\Vert\right)$ \vspace{.5cm}
		\item[(ii)] 
		$\displaystyle \left\Vert \exp \left( \sum_{j=1}^m A_j\right)-f_n(\{A_j\})\right \Vert	
		\leq
		\dfrac{2\cdot 3^m}{n}\left( \sum_{j=1}^m \Vert A_j\Vert\right)^2\exp\left(\frac{n+2}{n}\sum_{j=1}^m \Vert A_j\Vert\right)$
	\end{enumerate}
\end{theorem}
\begin{proof}
	Let
	\begin{align*}
		G&=\exp\left(\frac{1}{n}\sum_{j=1}^m A_j\right),\\
		H&=\left\{\exp\left(\frac{A_m}{2n}\right)\cdots\left\{\exp\left(\frac{A_2}{2n}\right)\exp\left(\frac{A_1}{n}\right)\exp\left(\frac{A_2}{2n}\right) \right\}\cdots	\exp\left(\frac{A_m}{2n}\right)\right\}.
	\end{align*}
	Then
	\begin{align*}
		\Vert G\Vert&\leq\exp \left(\frac{1}{n}\sum_{j=1}^m \Vert A_j\Vert\right), \nonumber \\ 
		\Vert H\Vert&\leq \exp\left(2\cdot\frac{\Vert A_m\Vert}{2n}\right)\cdots \exp\left(2\cdot \frac{\Vert A_2\Vert}{2n}\right)\exp\left(\frac{\Vert A_1\Vert}{n}\right)=\exp \left(\frac{1}{n}\sum_{j=1}^m \Vert A_j\Vert\right).
	\end{align*}
	
	Just as in the proof of Theorem \ref{Tjsuzuki1}
	\begin{align}
		t&=\left\Vert \exp \left( \sum_{j=1}^m A_j\right)-f_n(\{A_j\})\right \Vert	=\left\Vert G^n-H^n \right\Vert \nonumber \\
		&\leq\Vert G-H\Vert\left(\Vert G\Vert^{n-1}+\Vert G\Vert^{n-2}\Vert H\Vert+\cdots+\Vert H\Vert^{n-1}\right)\nonumber \\
		&\leq  n \Vert G-H\Vert\left(\max\{\Vert G\Vert, \Vert H\Vert\}  \right)^{n-1} \nonumber \\
		&\leq n \Vert G-H\Vert \exp \left(\frac{n-1}{n}\sum_{j=1}^m \Vert A_j\Vert\right)  \\
        & \leq n(\Vert G-F\Vert + \Vert H-F\Vert)\exp \left(\frac{n-1}{n}\sum_{j=1}^m \Vert A_j\Vert\right) \label{GHnbound}
	\end{align}
	
	Now we move to the proof of (i): for any $\displaystyle 2\leq k\leq m,$ recall
	\begin{align*}
		F_k=	I+\left(\frac{A_1}{n}+\frac{A_2}{n}+\cdots+\frac{A_k}{n}\right)+\dfrac{(A_1/n+A_2/n+\cdots+A_k/n)^2}{2!}.
	\end{align*}
	We claim that the degree 2 Taylor polynomials of $\displaystyle G$ and $\displaystyle H$ are $F_m,$ which is equal to $F$ defined in (\ref{NotionF}). It is trivial for $\displaystyle G$. For any $2\leq k\leq m,$ let
	\begin{align*}
		H_k=\left\{\exp\left(\frac{A_k}{2n}\right)\cdots\left\{\exp\left(\frac{A_2}{2n}\right)\exp\left(\frac{A_1}{n}\right)\exp\left(\frac{A_2}{2n}\right) \right\}\cdots	\exp\left(\frac{A_k}{2n}\right)\right\}.	
	\end{align*}
	A direct computation shows that 
	the degree 2 Taylor polynomial of $\displaystyle H_2$ is exactly
	the degree 2 Taylor polynomial of 
	\begin{align*}
		\left\{\left(I+\frac{A_2}{2n}+\frac{1}{2!}\left(\frac{A_2}{2n}\right)^2\right)\left(I+\frac{A_1}{n}+\frac{1}{2!}\left(\frac{A_1}{n}\right)^2\right)\left(I+\frac{A_2}{2n}+\frac{1}{2!}\left(\frac{A_2}{2n}\right)^2\right)\right\},	
	\end{align*} 
	which is
	\begin{align*}
		F_2=I+\left(\frac{A_1}{n}+\frac{A_2}{n}\right)+	\frac{1}{2!}\left(\frac{A_1}{n}+\frac{A_2}{n}\right)^2 .
	\end{align*}
	Similarly, the degree 2 Taylor polynomial of $\displaystyle H_3$ is exactly
	the degree 2 Taylor polynomial of 
	\begin{align*}
		\left\{\left(I+\frac{A_3}{2n}+\dfrac{\left(\frac{A_3}{2n}\right)^2}{2!}\right)F_2\left(I+\frac{A_3}{2n}+\dfrac{\left(\frac{A_3}{2n}\right)^2}{2!}\right)\right\}
	\end{align*} 
	which is
	\begin{align*}
		F_3=I+\left(\frac{A_1}{n}+\frac{A_2}{n}+\frac{A_3}{n}\right)+	\frac{1}{2!}\left(\frac{A_1}{n}+\frac{A_2}{n}+\frac{A_3}{n}\right)^2
	\end{align*}
	By induction, the degree 2 Taylor polynomial of $\displaystyle H$ is $F.$ 
	
	By (\ref{TaylorC}), 
	\begin{align}
		\left\Vert G-F\right \Vert\leq \frac{1}{3!\cdot n^3}\left(\sum_{j=1}^m \Vert A_j\Vert \right)^3\exp \left(\frac{1}{n}\sum_{j=1}^m \Vert A_j\Vert \right) \label{TaylorG},
	\end{align}
	since $\displaystyle G=C,$ which is defined (\ref{NotionC}).
	
	Since
	\begin{align*}
		H_2=2\left(\exp\left(\frac{A_2}{2n}\right)\circ \exp\left(\frac{A_1}{n}\right)\right)\circ \exp\left(\frac{A_2}{2n}\right)-\exp\left(\frac{A_2}{n}\right)\circ \exp\left(\frac{A_1}{n}\right),	
	\end{align*}
	it follows that for any integer $p>2,$ the norm of the sum of all terms of degree $p$ of Taylor expansion of $H_2\leq$ the sum of all terms of degree $p$ of Taylor expansion of \\ 
	$\displaystyle 3\cdot \exp\left(\Vert A_1\Vert/n\right)\exp\left(\Vert A_2\Vert/n\right).$
	
	Therefore, 
	\begin{align*}
		\left\Vert H_2-F_2\right\Vert&=\left\Vert \sum_{p>2}\left(\mbox{sum of all terms of degree $p$ of Taylor expansion of}\, H_2\right) \right\Vert \\
		&\leq 3\left[\exp \left(\frac{1}{n}\sum_{j=1}^2 \Vert A_j\Vert \right)-\left(I+\frac{1}{n}\sum_{j=1}^2 \Vert A_j\Vert+\frac{1}{2n^2}\left(\sum_{j=1}^2 \Vert A_j\Vert\right)^2 \right) \right]	
	\end{align*}
	which follows from the fact that the degree 2 Taylor polynomial of $\displaystyle H_2-F_2$ is zero.
	
	Note that 
	\begin{align*}
		H_3=2\left(\exp\left(\frac{A_3}{2n}\right)\circ H_2\right)\circ \exp\left(\frac{A_3}{2n}\right)-\exp\left(\frac{A_3}{n}\right)\circ H_2.	
	\end{align*}
	For any positive integer $p>2,$ the norm of the sum of all terms of degree $p$ of Taylor expansion of $H_3\leq$ the sum of all terms of degree $p$ of Taylor expansion of $\displaystyle 3\cdot 3\cdot \exp\left(\Vert A_1\Vert/n\right)\exp\left(\Vert A_2\Vert/n\right)\exp\left(\Vert A_3\Vert/n\right).$
	
	Therefore, 
	\begin{align*}
		\left\Vert H_3-F_3\right\Vert&= \left\Vert\sum_{p>2}\left(\mbox{sum of all terms of degree $p$ of Taylor expansion of}\, H_3\right) \right\Vert \\
		&\leq 3^2\left[\exp \left(\frac{1}{n}\sum_{j=1}^3 \Vert A_j\Vert \right)-\left(I+\frac{1}{n}\sum_{j=1}^3 \Vert A_j\Vert+\frac{1}{2n^2}\left(\sum_{j=1}^3 \Vert A_j\Vert\right)^2 \right) \right]	
	\end{align*}
	which follows from the fact that the degree 2 Taylor polynomial of $\displaystyle H_3-F_3$ is zero.
	
	By induction, 
	\begin{align}
		\left\Vert H_m-F_m\right\Vert &=\left\Vert H-F\right \Vert \nonumber \\
		&\leq 3^{m-1}\left[\exp \left(\frac{1}{n}\sum_{j=1}^m \Vert A_j\Vert \right)-\left(I+\frac{1}{n}\sum_{j=1}^m \Vert A_j\Vert+\frac{1}{2n^2}\left(\sum_{j=1}^m \Vert A_j\Vert\right)^2 \right) \right]\nonumber \\
		&\leq 3^{m-1} \frac{1}{3!\cdot n^3}\left(\sum_{j=1}^m \Vert A_j\Vert \right)^3\exp \left(\frac{1}{n}\sum_{j=1}^m \Vert A_j\Vert \right) \label{TaylorH}
	\end{align}
	where the inequality (\ref{TaylorH}) follows from \cite[Theorem 1]{Suzuki1976}.

	
	From (\ref{GHnbound}),(\ref{TaylorG}) and (\ref{TaylorH}) we obtain 
	\begin{align*}
		\left\Vert \exp \left( \sum_{j=1}^m A_j\right)-f_n(\{A_j\})\right \Vert	
		\leq
		\dfrac{3^{m-1}+1}{6n^2}\left( \sum_{j=1}^m \Vert A_j\Vert\right)^3\exp\left(\sum_{j=1}^m \Vert A_j\Vert\right).
	\end{align*}

	We turn to the proof of (ii):
	\begin{align}
		\Vert G-H \Vert&=\left\Vert G-U_{\exp(\frac{A_m}{2n})}U_{\exp(\frac{A_{m-1}}{2n})}\cdots U_{\exp(\frac{A_2}{2n})}U_{\exp(\frac{A_1}{2n})}\left(1\right)\right \Vert \nonumber \\ 
		&=\left\Vert U_{\exp(\frac{A_m}{2n})}U_{\exp(\frac{A_{m-1}}{2n})}\cdots U_{\exp(\frac{A_2}{2n})}U_{\exp(\frac{A_1}{2n})}\left(U(G)-1\right)\right \Vert \nonumber \\
		&\leq  \exp \left(\frac{1}{n}\sum_{j=1}^m \Vert A_j\Vert\right)\cdot \left\Vert \left(U(G)-1\right)\right \Vert , \label{GHbound}
	\end{align}
	where $\displaystyle U=U_{\exp\left(-\frac{A_1}{2n}\right)}U_{\exp\left(-\frac{A_{2}}{2n}\right)}\cdots U_{\exp\left(-\frac{A_m}{2n}\right)}.$ 
	
	Following similar argument as in (i) and using the fact that the Taylor polynomial of degree 1 of $\displaystyle U(G)-I$ is zero,
	\begin{align} \label{UGbound}
		\left\Vert U(G)-1\right \Vert &\leq 3^m\exp \left(\frac{2}{n}\sum_{j=1}^m \Vert A_j\Vert \right)-\left(1+\frac{2}{n}\sum_{j=1}^m \Vert A_j\Vert \right) \\ 
		&\leq 3^m \frac{2}{n^2}\left(\sum_{j=1}^m \Vert A_j\Vert \right)^2\exp \left(\frac{2}{n}\sum_{j=1}^m \Vert A_j\Vert \right) \label{Taylor}          
	\end{align}
	where the inequality (\ref{Taylor}) follows from \cite[Theorem 1]{Suzuki1976}.
	
	Combining (\ref{GHnbound}), (\ref{GHbound}) and (\ref{Taylor})  
	\begin{align*}
		\left\Vert \exp \left( \sum_{j=1}^m A_j\right)-f_n(\{A_j\})\right \Vert	
		\leq
		\dfrac{2\cdot 3^m}{n}\left( \sum_{j=1}^m \Vert A_j\Vert\right)^2\exp\left(\frac{n+2}{n}\sum_{j=1}^m \Vert A_j\Vert\right).
	\end{align*}
	
\end{proof}
\begin{remark}
If $\Al$ is a Jordan subalgebra of $B(H)_{\rm sa},$ then the inequalities in Theorem \ref{Tjsuzuki2} are reduced to the classical Suzuki estimates \cite{Suzuki1976, Suzuki1985}:
\begin{enumerate}
	\item[(i)] $\displaystyle \left\Vert \exp \left( \sum_{j=1}^m A_j\right)-f_n(\{A_j\})\right \Vert	
	\leq
	\dfrac{1}{3n^2}\left( \sum_{j=1}^m \Vert A_j\Vert\right)^3\exp\left(\sum_{j=1}^m \Vert A_j\Vert\right),$ \vspace{0.5cm}
	\item[(ii)] 
	$\displaystyle \left\Vert \exp \left( \sum_{j=1}^m A_j\right)-f_n(\{A_j\})\right \Vert	
	\leq
	\dfrac{2}{n}\left( \sum_{j=1}^m \Vert A_j\Vert\right)^2\exp\left(\frac{n+2}{n}\sum_{j=1}^m \Vert A_j\Vert\right).$
\end{enumerate} 	
\end{remark}



As a corollary, we obtain the following generalization of 
symmetrized Lie-Trotter formula for an 
arbitrary number of elements in a JB-algebra. 

\begin{corollary}
	For any finite number of elements $A_1, A_2,\cdots, A_m$ in a JB-algebra $\Al,$	
	\begin{align*}
		\lim\limits_{n\to \infty}f_n\left(\{A_j\}\right)=\exp\left(\sum_{j=1}^m A_j\right).	
	\end{align*}
\end{corollary}



\section{Generalized Trotter formulae in Jordan-Banach algebras}\label{section: general}
Motivated by the formidable Theorem 2.1 in \cite{Peralta2023}, we establish three generalized Lie-Trotter formulas for an arbitrary finite number of elements in a Jordan-Banach algebra. 
\begin{proposition}	\label{PJBLie1}
	Let $A_1, A_2,\cdots, A_m$ be elements in  a unital Jordan-Banach algebra $\Al.$ Then,	\begin{align}\label{STformula2}
		\lim\limits_{n\to \infty}g_n\left(\{A_j\}\right)=\exp\left(\sum_{j=1}^m A_j\right).
	\end{align}
\end{proposition}
\begin{proof}
	For $2\leq k\leq m,$ let
	\begin{align*}
		g_k(z)=\left[\left(\exp\left(z A_1 \right)\circ \exp\left(z A_2 \right)\right)\circ\cdots \right]\circ \exp\left(z A_k \right).
	\end{align*}
	Note that for any $\displaystyle g_k: \ball{0, r}\to \Al$ is a holomorphic map satisfying $g_k(0)=1$ for any $\displaystyle r>0$ and $ \displaystyle g_{k+1}(z)=g_k(z)\circ \exp\left(z A_k \right).$ It is easy to see that $\displaystyle g_2'(0)=A_1+A_2.$ Thus,
	\begin{align*}
		g_3'(0)=g_2'(0)\circ \exp\left(0\cdot A_3 \right)+g_2(0)\circ \left(A_3 \circ \exp\left(0 \cdot A_3 \right)\right)=A_1+A_2+A_3. 	
	\end{align*}
	By induction
	\begin{align*}
		g_k'(0)=A_1+A_2+\cdot+A_k.	
	\end{align*}
	Applying \cite[Theorem 2.1]{Peralta2023} with $\displaystyle \lambda_n=\frac{1}{n}$ and $\mu_n=n,$ \begin{align*}
		\lim\limits_{n\to \infty}\left(g_m\left(\frac{1}{n}\right)\right)^n=\exp \left( g_m'(0)\right),
	\end{align*}
	which is
	\begin{align*}
		\lim\limits_{n\to \infty}\left(\left[\left(\exp\left(\frac{A_1}{n}\right)\circ \exp\left(\frac{A_2}{n}\right)\right)\circ\cdots \right]\circ \exp\left(\frac{A_m}{n}\right)\right)^n=\exp\left(\sum_{j=1}^m A_j\right).
	\end{align*}
\end{proof}


The following proposition is another generalized Lie-Trotter formula.

\begin{proposition}
	For any finite number of elements $A_1, A_2,\cdots, A_{2m+1}$ in a unital Jordan-Banach algebra $\Al,$
	\begin{align}
		\exp \left( \sum_{j=1}^{2m+1} A_j\right)=\lim_{n\to \infty}h_n(\{A_j\})
	\end{align}
	where
	\begin{align*}
		h_n(\{A_j\})=\left\{\exp\left(\frac{A_{2m}}{n}\right)\cdots\left\{\exp\left(\frac{A_2}{n}\right)\exp\left(\frac{A_1}{n}\right)\exp\left(\frac{A_3}{n}\right) \right\}\cdots	\exp\left(\frac{A_{2m+1}}{n}\right)\right\}^n	
	\end{align*}	
\end{proposition}

\begin{proof}
	For $1\leq k\leq m,$ denote
	\begin{align*}
		h_k(z)=\left\{\exp\left(z A_{2k}\right)\cdots\left\{\exp\left(z A_2\right)\exp\left(z A_1\right)\exp\left(z A_3 \right) \right\}\cdots	\exp\left(z A_{2k+1} \right)\right\}	
	\end{align*}
	For any $r>0,$ $f_k: \ball{0, r}\to \Al$ is a holomorphic map satisfying $h_k(0)=1.$ 
	
	Note that 
	\begin{align*}
		h_1'(z)=\frac{d}{dz}\left[(e^{z A_2}\circ e^{z A_1})\circ e^{z A_3}\right]+\frac{d}{dz}\left[(e^{z A_3}\circ e^{z A_1})\circ e^{z A_2}\right]-\frac{d}{dz}\left[(e^{z A_2}\circ e^{z A_3})\circ e^{z A_1}\right]	
	\end{align*}
	By the proof of Proposition \ref{PJBLie1},  
	\begin{align*}
		\frac{d}{dz}\vert_{z=0}\left[(e^{z A_2}\circ e^{z A_1})\circ e^{z A_3}\right]&=\frac{d}{dz}\vert_{z=0}\left[(e^{z A_3}\circ e^{z A_1})\circ e^{z A_2}\right]\\
		&=\frac{d}{dz}\vert_{z=0}\left[(e^{z A_2}\circ e^{z A_3})\circ e^{z A_1}\right]\\
		&=A_1+A_2+A_3.	
	\end{align*}
	Therefore, 
	\begin{align*}
		h_1'(0)=A_1+A_2+A_3.	
	\end{align*}
	Since $\displaystyle h_2(z)=\left\{e^{z A_4} h_1(z) e^{z A_5}\right\},$
	\begin{align*}
		h_2'(z)&=\left\{\dfrac{ d \left(e^{z A_4}\right)}{dz} h_1(z) e^{z A_5}\right\}+\left\{e^{z A_4} h_1'(z) e^{z A_5}\right\}+\left\{e^{z A_4} h_1'(z) \dfrac{d\left(e^{z A_5}\right)}{dz}\right\}.
	\end{align*}
	Thus, 
	\begin{align*}
		h_2'(0)=A_4+A_1+A_2+A_3+A_5.	
	\end{align*}
	By induction, 
	\begin{align*}
		h_k'(0)=\sum_{j=1}^{2k+1}A_j.	
	\end{align*}
	Applying \cite[Theorem 2.1]{Peralta2023} with $\displaystyle \lambda_n=\frac{1}{n}$ and $\mu_n=n,$ 
	\begin{align*}
		\lim\limits_{n\to \infty}\left(h_m\left(\frac{1}{n}\right)\right)^n=\exp \left( h_m'(0)\right).	
	\end{align*}
	i.e.
	\begin{align*}
		\lim\limits_{n\to \infty}h_n\left(\{A_j\}\right)=\exp\left(\sum_{j=1}^{2m+1} A_j\right).
	\end{align*}
\end{proof}
As a consequence, we get the following generalized Lie-Trotter formula for JB-algebras, which are over $\Rdb.$

\begin{corollary}
	For any finite number of elements $A_1, A_2,\cdots, A_{2m+1}$ in a JB-algebra $\Al,$
	\begin{align}
		\exp \left( \sum_{j=1}^{2m+1} A_j\right)=\lim_{n\to \infty}h_n(\{A_j\})
	\end{align}
\end{corollary}
\begin{proof}
	The complexification of $\Al$ is a JB*-algebra (see \cite{Wright1977}), thus a complex Jordan-Banach algebra.	
\end{proof}



The following result is a generalization of symmetrized Lie-Trotter formula for an arbitrary number of elements in a Jordan-Banach algebra.

\begin{corollary}
	Let $\Al$ be a unital Jordan-Banach algebra. Then the formula
	\begin{align}\label{STformula}
		\lim\limits_{n\to \infty}f_n\left(\{A_j\}\right)=\exp\left(\sum_{j=1}^m A_j\right)
	\end{align}
	holds for all $A_1, A_2,\cdots, A_m\in \Al$ with an arbitrary positive integer $m.$	
\end{corollary}

\bigskip
\noindent 
{\bf Acknowledgement:}\ S. C. acknowledges DOE ASCR funding under the Quantum Computing Application Teams program, FWP number ERKJ347. This manuscript has been authored in part by UT-Battelle, LLC, under Contract No. DE-AC0500OR22725 with the U.S. Department of Energy. The United States Government retains and the publisher, by accepting the article for publication, acknowledges that the United States Government retains a non-exclusive, paid-up, irrevocable, worldwide license to publish or reproduce the published form of this manuscript, or allow others to do so, for the United States Government purposes. The Department of Energy will provide public access to these results of federally sponsored research in accordance with the DOE Public Access Plan.

The authors would like to thank the referee for the careful reading of the manuscript and valuable comments.  

\bigskip
\noindent 
{\bf Declaration of competing interest:}\
The authors declare that they have no known competing financial interests or personal relationships that could have appeared to influence the work reported in this paper.

\bigskip
\noindent 
{\bf Data availability:} \
Data sharing is not applicable to this article as no new data were created or analyzed in this study.

\end{document}